\documentclass{article}

\usepackage[preprint]{neurips_2026}

\usepackage[utf8]{inputenc} 
\usepackage[T1]{fontenc}    
\usepackage{hyperref}       
\usepackage{url}            
\usepackage{booktabs}       
\usepackage{amsfonts}       
\usepackage{nicefrac}       
\usepackage{microtype}      
\usepackage{xcolor}         
\usepackage{amsmath,amsthm}
\usepackage{graphicx}

\usepackage{multirow}
\usepackage{makecell}

\newtheorem{theorem}{Theorem}

\newtheorem{assumption}{Assumption}

\title{Approximating full conformal prediction: distribution free guarantees via the tournament correction}

\author{
  Aabesh Bhattacharyya \\
  Department of Statistics\\
  University of Chicago\\
  Chicago, IL 60637 \\
  \texttt{aabesh@uchicago.edu} \\
  \And
  Boxuan Zhang\\
  Department of Statistics\\
  University of Chicago\\
  Chicago, IL 60637\\
  \texttt{boxuanzhang@uchicago.edu} \\
  \AND
  Rina Foygel Barber \\
  Department of Statistics\\
  University of Chicago\\
  Chicago, IL 60637\\
  \texttt{rina@uchicago.edu} \\
}

\begin{document}

\maketitle

\begin{abstract}

Conformal prediction is a framework for providing prediction intervals with distribution-free validity, guaranteeing predictive coverage for data drawn from any distribution. Its two main variants are full conformal prediction and split conformal prediction (also called transductive and inductive).  Full conformal prediction is widely considered to be statistically more efficient (since split conformal prediction requires data splitting, and therefore can lead to wider prediction intervals due to the resulting loss in sample size), but  its implementation is computationally prohibitive, as it requires the underlying model to be refit for every candidate value in the response space. Existing computational shortcuts, such as using a discrete grid of values to approximate the full conformal prediction construction, frequently lack theoretical guarantees on marginal coverage and can fail in practice. 

To address this limitation, we introduce a novel class of approximations to the full conformal prediction method, based on the idea of \emph{tournaments}, which enables the construction of prediction sets with a rigorous marginal coverage guarantee of $1-2\alpha$. Under stability conditions, the theoretical coverage guarantee tightens to approximately $1-\alpha$. This new framework generalizes the existing method of leave-one-out cross-conformal prediction, while allowing for flexible use of various existing approximation strategies.
 
\end{abstract}

\section{Introduction}\label{sec:introduction}

Conformal prediction is a general and flexible framework for uncertainty quantification that converts predictions from machine learning or other black-box models into prediction sets with finite-sample validity. Its model-agnostic nature and minimal assumptions have made it useful in a wide range of applications, including medicine \citep{vazquez2022conformal,lu2022fair} and finance \citep{wisniewski2020application}. As increasingly complex prediction algorithms are being deployed in practice, such distribution-free tools for assessing uncertainty have become especially important.

Formally, let $Z_i=(X_i,Y_i) \in \mathcal X \times \mathcal Y$ for $i=1,\dots,n+1$, and suppose that $Z_1,\dots,Z_n,Z_{n+1}$ are exchangeable, with $Y_{n+1}$ unobserved at test time. Given a target miscoverage level $\alpha \in (0,1)$, the goal is to construct a prediction set $\widehat C_{n,\alpha}(X_{n+1})$ such that
\begin{equation}\label{eqn:marginal}
\mathbb P\bigl(Y_{n+1} \in \widehat C_{n,\alpha}(X_{n+1})\bigr) \ge 1-\alpha.
\end{equation}

In \emph{full conformal prediction} \citep{vovk2005algorithmic}, for each candidate label $y \in \mathcal Y$, one forms the augmented dataset $\mathcal D_{n+1}^y = \{Z_1,\dots,Z_n,Z_{n+1}^y\}$ where $Z_{n+1}^y = (X_{n+1},y)$ and the prediction set is defined as
\begin{align}\label{eq: full_cp}
    \widehat{C}_{n,\alpha}^{\mathrm{full}}(X_{n+1}) =
\left\{
y \in \mathcal Y :
s(Z_{n+1}^y;\mathcal{D}_{n+1}^y) \leq Q_{\left(1-\alpha\right)\left(1+\frac{1}{n}\right)}\left(\left\{s(Z_i,\mathcal{D}_{n+1}^y)\right\}_{i=1}^n\right)
\right\}.
\end{align}
Here $s(z;\mathcal{D})$ is a \emph{score function}, with large values indicating that $z$ appears unusual (does not `conform') to the trends observed in $\mathcal{D}$, while $Q_\tau(\cdot)$ denotes the $\tau$-quantile of a list of values. The function $s(z;\mathcal D)$ is required to be \emph{symmetric}, meaning that it is invariant to reordering the data points in $\mathcal D$. For intuition, we can consider a canonical example: the \emph{residual score}, given by
\begin{equation}\label{eqn:resid_score}s(z;\mathcal{D}) = |y  - \hat{f}_{\mathcal{D}}(x)|,\end{equation}
where $\hat{f}_{\mathcal{D}}$ denotes a fitted model produced by a (symmetric) regression algorithm trained on the dataset $\mathcal{D}$. Full conformal prediction offers a distribution-free guarantee of predictive coverage: if the training and test data points $Z_1,\dots,Z_{n+1}$ are exchangeable, then for any symmetric score function $s$, the marginal coverage guarantee~\eqref{eqn:marginal} is guaranteed to hold \citep{vovk2005algorithmic}.

An important aspect of conformal prediction is the trade-off between statistical efficiency and computational cost. In particular,
full conformal suffers from a high computational cost: one must refit or reevaluate the underlying model on the augmented dataset $\mathcal D_{n+1}^y$ for every possible value $y$, which is often infeasible (e.g., for a real-valued response, when $\mathcal Y = \mathbb R$).
A popular alternative is \emph{split conformal prediction} \citep{papadopoulos2008inductive}. Here the sample is divided into a training set $\mathcal D_{\mathrm{tr}}$ and a calibration set $\mathcal D_{\mathrm{cal}}$, each of size $n/2$. Then the prediction interval is defined as
\[
\widehat C_{n,\alpha}^{\mathrm{split}}(X_{n+1})
=
\left\{y \in \mathcal Y :
s((X_{n+1},y);\mathcal D_{\mathrm{tr}}) \leq Q_{\left(1-\alpha\right)\left(1+\frac{1}{n/2}\right)}\left(\left\{s\left(Z_i;\mathcal{D}_{\mathrm{tr}}\right)
\right\}_{i:Z_i \in \mathcal{D}_{\mathrm{cal}}}\right)\right\}.
\]
This procedure is computationally attractive since it requires only a single model fit, but it pays a statistical price for sample splitting---only half of the data is used to train the predictor, and the other half is used to calibrate the interval width. As a result, split conformal prediction often produces wider prediction sets than full conformal.

Motivated by this tension, many practical works consider approximations to full conformal that aim to retain its statistical benefits without its computational burden. Common examples include evaluating full conformal only on a grid of candidate values $y$, or using a naive plug-in score in which the model is fit once on the training data and then reused for all candidate responses. Although such methods can be effective in practice, they typically do not come with general finite-sample coverage guarantees and can fail badly in certain situations.

\paragraph{Our contribution.}
We introduce a framework for modifying approximations to the full conformal prediction set, using a \emph{tournament correction}, that gives rise to a broad family of conformal prediction methods that are implementable in practice and enjoy a finite-sample coverage guarantee of $1-2\alpha$.  In addition, under suitable stability assumptions, we show that a small inflation of the resulting interval can further improve the guarantee, yielding coverage close to $1-\alpha$. This framework includes the existing leave-one-out cross-conformal procedure as a special case.

\section{Preliminaries: approximations to full conformal}\label{sec:background}
In this section, we summarize several common approaches towards approximating full conformal prediction, in the setting where $\mathcal{Y}$ is large (e.g., a real-valued response, $\mathcal Y = \mathbb R$) so that computing $\widehat{C}_{n,\alpha}^{\mathrm{full}}(X_{n+1})$ is infeasible.

\paragraph{Example 0: Deletion.}
We first consider a baseline method, which neglects to use conformal prediction entirely: we train a model on the training dataset $\mathcal{D}_n$ only, and use the training residuals or scores to form a prediction set. To compare to full conformal~\eqref{eq: full_cp}, we can think of this approach as taking an approximation via deletion: we simply delete (or rather, fail to include) the test point into the model training process.

That is, instead of computing the score function using the augmented dataset $\mathcal{D}^y_{n+1}$, we use only the training points $\mathcal D_n$, leading to the prediction set
    \[\widehat{C}_{n,\alpha}^{\mathrm{delete}}(X_{n+1}) = \left\{y\in\mathcal Y : s(Z^y_{n+1};\mathcal D_n) \leq Q_{(1-\alpha)(1+\frac{1}{n})}(\{s(Z_i;\mathcal D_n)\}_{i=1}^n)\right\}.\]

Of course, this above approach will often lead to severe undercoverage, since we are ignoring the issue of overfitting. For example, if we use the residual score~\eqref{eqn:resid_score}, we obtain the prediction set
\[\widehat{C}_{n,\alpha}^{\mathrm{delete}}(X_{n+1}) = \hat{f}_{\mathcal D_n}(X_{n+1}) \pm   Q_{(1-\alpha)(1+\frac{1}{n})}(\{|Y_i - \hat{f}_{\mathcal D_n}(X_i)|\}_{i=1}^n),\]
which is likely much too narrow since the training residuals $|Y_i - \hat{f}_{\mathcal D_n}(X_i)|$ are typically smaller than the test point's residual $|Y_{n+1} - \hat{f}_{\mathcal D_n}(X_{n+1})|$.
Nonetheless, we can think of this as a (typically very inaccurate) `approximation' to full conformal prediction.

\paragraph{Example 1: Rounding.} In practice, the most commonly used approximation for full conformal prediction is to simply use rounding: that is, we consider a discrete grid of $y$ values and only fit the model at those values \citep{lei2018distribution}.

Concretely, instead of computing the score function using the augmented dataset $\mathcal{D}^y_{n+1}$, we replace $y$ with its rounded value. Writing $[y]$ to denote the operation of rounding $y$ to a finite grid of values $y^{(1)},\dots,y^{(M)}\in\mathcal Y$, we then have the prediction set
    \[\widehat{C}_{n,\alpha}^{\mathrm{round}}(X_{n+1}) = \left\{y\in\mathcal Y : s(Z^y_{n+1};\mathcal D^{[y]}_{n+1}) \leq Q_{(1-\alpha)(1+\frac{1}{n})}(\{s(Z_i;\mathcal D^{[y]}_{n+1})\}_{i=1}^n)\right\}.\]

Note that, in this approach, we only need to fit a finite number of models: for instance, for the residual score~\eqref{eqn:resid_score}, we only need to train $M$ models,
$\hat{f}_m = \hat{f}_{\mathcal{D}_{n+1}^{y^{(m)}}}$ for each $ m=1,\dots,M$,
and can then compute the prediction set as
\begin{multline*}\widehat{C}_{n,\alpha}^{\mathrm{round}}(X_{n+1}) = \bigcup_{m=1,\dots,M}\bigg\{y\in\mathcal Y : [y] = y^{(m)}, \\ |y - \hat{f}_m(X_{n+1})| \leq Q_{(1-\alpha)(1+\frac{1}{n})}(\{|Y_i - \hat{f}_m(X_i)|\}_{i=1}^n)\bigg\}.\end{multline*}

\paragraph{Example 2: One-step update.} In settings where the fitted model is obtained by solving an optimization problem, another common approach is the one-step update: for example, if $\hat{f}_{\mathcal D}$ is obtained by optimizing some penalized likelihood method, we might train the model on $\mathcal D_n$, and then add in the hypothesized test point $(X_{n+1},y)$ by taking one step of gradient descent. 
More generally, we approximate the full conformal set~\eqref{eq: full_cp} by adding in the test point approximately through a one-step update \citep{martinez2023approximating,tailor2025approximating}: we write $\mathrm{Update}(\hat{f}_{\mathcal D},z)$ to denote that a trained model $\hat{f}_{\mathcal D}$ is then updated approximately with a new data point $z$. 

For example, if we assume a parametric predictor $f_\theta:\mathcal{X}\to\mathbb{R}$ and
$\widehat f_{\mathcal{D}} = f_{\widehat\theta(\mathcal{D})}$ where
\[
\widehat\theta(\mathcal{D}) := \arg\min_{\theta} \; L_{\mathcal{D}}(\theta).
\]
Then $\mathrm{Update}(\hat f_{\mathcal{D}},z) = f_{\widetilde \theta}$ where $\widetilde\theta$ is obtained by taking one gradient descent step:
\[
\tilde\theta
:=
\hat\theta(\mathcal{D}) - \eta \,\nabla L_{\mathcal{D} \cup \{z\}}(\hat\theta(\mathcal{D})),
\]
assuming that the loss is differentiable (here $\eta>0$ is some step size parameter). Let $\ell(z,\hat{f})$ denote any loss that compares a data point $z$ against a fitted model $\hat{f}$, e.g., we might choose the absolute residual score $\ell(z,\hat{f})=|y-\hat{f}(x)|$ where $z=(z,y)$. We then have the prediction set
    \begin{multline*}\widehat{C}_{n,\alpha}^{\mathrm{one-step}}(X_{n+1}) = \bigg\{y\in\mathcal Y : \ell\left(Z_{n+1}^y, \mathrm{Update}(\hat{f}_{\mathcal D_n},Z_{n+1}^y)\right) \leq \\Q_{(1-\alpha)(1+\frac{1}{n})}\left(\left\{\ell\left(Z_i, \mathrm{Update}(\hat{f}_{\mathcal D_n},Z_{n+1}^y)\right)\right\}_{i=1}^n\right)\bigg\}.\end{multline*}

\paragraph{Example 3: Approximating the Bayesian PPD.} Finally, we consider a Bayesian setting, where our aim is to use the posterior predictive distribution (PPD) for the score. 
Given a prior $\pi$ on parameter $\theta$, and a likelihood $f_\theta(y\mid x)$ for the conditional distribution of $Y\mid X$, the PPD is given by
\[\mathbb{E}_{\theta\sim \pi(\cdot\mid\mathcal{D})}[f_\theta(y\mid x)],\]
i.e., the expected likelihood when $\theta$ is sampled from the posterior, given some dataset $\mathcal{D}$. Since we would like to retain values of $Y$ that are likely under the PPD, we would ideally like to use the score $s((x,y);\mathcal D) = -\mathbb{E}_{\theta\sim \pi(\cdot\mid\mathcal{D})}[f_\theta(y\mid x)]$. In practice, the PPD is generally estimated with a Monte Carlo approximation, and
\[s((x,y);\mathcal{D}) = -\frac{1}{K}\sum_{k=1}^K f_{\theta_k}(y\mid x)\textnormal{ where }\theta_1,\dots,\theta_K\sim \pi(\cdot\mid\mathcal{D}).\]
However, it is impractical to compute scores $s(\cdot;\mathcal{D}_{n+1}^y)$ for all values of $y$ (since we cannot draw samples from each possible posterior distribution). This motivates an `add-one-in' (AOI) approximation to the PPD, using importance weighting \citep{fong2021conformal,deliu2025interplay}: after sampling $\theta_k\sim \pi(\cdot\mid\mathcal{D}_n)$, the appropriate importance weight for adding $(X_{n+1},y)$ into the training set is proportional to $f_{\theta_k}(y\mid X_{n+1})$, and we therefore
have the prediction set
    \begin{multline*}\widehat{C}_{n,\alpha}^{\mathrm{AOI-PPD}}(X_{n+1}) = \bigg\{y\in\mathcal Y : -\sum_{k=1}^K f_{\theta_k}(y\mid X_{n+1}) \cdot f_{\theta_k}(y\mid X_{n+1}) \leq \\Q_{(1-\alpha)(1+\frac{1}{n})}\left(\left\{-\sum_{k=1}^Kf_{\theta_k}(y\mid X_{n+1}) \cdot f_{\theta_k}(Y_i\mid X_i)\right\}_{i=1}^n\right)\bigg\},\end{multline*}
    where $\theta_1,\dots,\theta_K\sim\pi(\cdot\mid \mathcal{D}_n)$.

\subsection{Unifying notation} 
To summarize these examples, we introduce a common notation that encompasses all of these various approximations to full conformal prediction. Write
\[s_{\mathrm{approx}}(z;\mathcal D + \{z'\})\]
which computes a score for data point $z$, as compared to a dataset $\mathcal D$ and an additional data point $z'$. The idea is that this is an approximation to the score computed on a dataset that includes $z'$,
\[s_{\mathrm{approx}}(z;\mathcal D + \{z'\}) \approx s(z;\mathcal{D}\cup \{z'\}).\]
However, $s_{\mathrm{approx}}(z;\mathcal D + \{z'\}$ is required to be symmetric in $\mathcal D$ but \emph{not} in the combined dataset $\mathcal D\cup\{z'\}$.
We then define a prediction set
\begin{equation}\label{eq:approximate_conformal_sets}\begin{split}
    \widehat{C}_{n,\alpha}^{\mathrm{approx}}(X_{n+1}) = 
\Big\{
y \in \mathcal Y :&
s_{\mathrm{approx}}(Z_{n+1}^y;\mathcal{D}_n + \{Z_{n+1}^y\}) \leq \\&Q_{\left(1-\alpha\right)\left(1+\frac{1}{n}\right)}\left(\left\{s_{\mathrm{approx}}(Z_i;\mathcal{D}_n + \{Z_{n+1}^y\})\right\}_{i=1}^n\right)
\Big\}.\end{split}
\end{equation}
Each of the four examples described above can be expressed as an instance of this unified notation, by choosing the score $s_{\mathrm{approx}}$ as the following:
\begin{itemize}
    \item \textbf{Example 0: Deletion.} We define $s_{\mathrm{approx}}(z;\mathcal D + \{z'\}) = s(z;\mathcal D)$, i.e., the data point $z'$ is not included when computing the score. 
    \item \textbf{Example 1: Rounding.} Let $z'=(x',y')$, and let $[y']$ denote the value of $y'$ after rounding to the grid. Then $s_{\mathrm{approx}}(z;\mathcal D + \{z'\}) = s(z;\mathcal D \cup\{(x',[y'])\})$, i.e., the response $y'$ is rounded to $[y']$ when computing the score.
    \item \textbf{Example 2: One-step update.} Writing $\mathrm{Update}(\cdot)$ to denote a one-step model update rule (e.g., gradient descent) as before, we have
    $s_{\textnormal{approx}}(z;\mathcal{D}+\{z'\}) = \ell\big(z,\mathrm{Update}(\hat{f}_{\mathcal D},z')\big)$.
    \item \textbf{Example 3: Approximating the Bayesian PPD.}\footnote{We remark that this example uses a \emph{randomized} score, since the score function depends on random samples drawn from the posterior. Conformal prediction theory extends naturally to the setting of randomized scores \citep{vovk2005algorithmic}. The results of our work hold for this setting as well, as we will show in the Appendix.} Let $z=(x,y)$ and $z'=(x',y')$. We then define $s_{\mathrm{approx}}(z;\mathcal{D}+\{z'\}) = -\sum_{k=1}^K f_{\theta_k}(y\mid x)\cdot f_{\theta_k}(y'\mid x')$, where $\theta_1,\dots,\theta_K\sim\pi(\cdot\mid \mathcal{D})$.
\end{itemize}

\subsection{Other related works}
Beyond the four classes of methods considered here, prior work on approximating full conformal prediction has mainly taken two forms: exploiting model-specific structure to compute full conformal sets more efficiently, and proving validity of cheaper approximations under algorithmic stability. For example, \citet{lei2019fast}, \citet{ndiaye2019computing}, \citet{ndiaye2023root}, and \citet{cherubin2021exact} develop efficient exact or approximate full conformal algorithms in structured settings such as penalized regression, nearest-neighbor, and kernel methods. Relatedly, cross-conformal prediction \citep{vovk2015cross} and the jackknife+/CV+ methods of \citet{barber2021predictive} are methods closely related to ours; as we will see below, leave-one-out cross conformal is a special case of our framework. More recently, \citet{ndiaye2022stable} and \citet{lee2025leave} study computationally efficient conformal procedures under stability assumptions. These works are complementary to ours: they recover validity through stability assumptions, whereas our tournament based corrected methods achieve finite-sample validity without requiring any assumptions, with stability used only later to sharpen the guarantee.

\section{Tournament correction}\label{sec:tourn_correction}

In this section we formally introduce our method. We mainly discuss how a slight modification to the score function $s_{\mathrm{approx}}$ can be used to develop methods that automatically provide $1-2\alpha$ coverage guarantees. At a high level, the idea is that we will now treat \emph{two} data points approximately in the training process---the test point $n+1$ along with one training point $i$, rather than treating only the test point approximately as in the examples above. Extending our earlier notation, we consider an approximate score of the form

\begin{align}\label{eq:tournament_score}
    s_{\mathrm{approx}}(z;\mathcal{D} + \{z',z''\})
\end{align}
which computes the score at data point $z$ when the dataset $\mathcal{D}$ and the data points $z',z''$ are used to train the model possibly in asymmetric ways. As before, we require $s_{\mathrm{approx}}(z;\mathcal{D} + \{z',z''\})$ to be symmetric in $\mathcal{D}$, and now also require $s_{\mathrm{approx}}(z;\mathcal{D} + \{z',z''\}) = s_{\mathrm{approx}}(z;\mathcal{D}+\{z'',z'\})$. 

Using this score we can define a prediction set
\begin{equation}\begin{split}
     \widehat C_{n,\alpha}^{\mathrm{tourn}}(X_{n+1}) := \Bigg\{y:\sum_{i=1}^n\mathbf{1}\Big\{& s_{\mathrm{approx}}(Z_{n+1}^y;\mathcal{D}_{n\setminus i} + \{Z_i,Z_{n+1}^y\})  \\ 
     & > s_{\mathrm{approx}}(Z_i;\mathcal{D}_{n\setminus i} + \{Z_i,Z_{n+1}^y\})\Big\} <(1-\alpha)(n+1)\Bigg\},
\end{split}\label{eq:tournament_corrected_methods}\end{equation}
where $\mathcal{D}_{n\setminus i } = \{Z_1,\dots,Z_{i-1},Z_{i+1},\dots,Z_{n}\}$, i.e., all the training data points with the $i$-th point removed. As compared to an approximation to full conformal (any approximation of the form given in~\eqref{eq:approximate_conformal_sets}), this general framework provides a modification to this approximation, which we refer to as the \emph{tournament correction} (we will explain the idea behind this name, below). 

\subsection{Examples}\label{sec:tournament_corrected_examples}

Next, we describe how each of the approximate methods stated in Section~\ref{sec:background} can be modified using this framework. 
\begin{itemize}
    \item \textbf{Example 0: Deletion.} $s_{\mathrm{approx}}(z;\mathcal{D}+\{z',z''\}) = s(z;\mathcal{D})$ i.e., we remove both the data points $z'$ and $z''$. The key difference from the approximate deletion method is that in the approximate method, to find the score at the training and the test point only the test point was removed and this introduced asymmetry in the training and the test scores, whereas, now when we compare the scores at $Z_{n+1}^y$ and $Z_i$ we remove both the $i$-th point and the test point. This method is exactly equivalent to the leave-one-out cross-conformal method proposed by \citet{vovk2015cross,vovk2018cross}.
    \item \textbf{Example 1: Rounding.} If $z' = (x',y')$ and $z'' = (x'',y'')$, define $s_{\mathrm{approx}}(z;\mathcal{D}+\{z',z''\}) = s(z;\mathcal{D} \cup \{(x',[y']) \cup (x'',[y''])\})$. Again in this case, the difference from the approximate method is that in the tournament-corrected method~\eqref{eq:tournament_corrected_methods}, both the training point $Z_i$ and the hypothesized test point $Z^y_{n+1}$ are rounded, when computing the scores.
    \item \textbf{Example 2: One-step update.} As before, writing $\mathrm{Update}(\cdot)$ to denote a one-step model update rule (e.g., gradient descent), we have
    $s_{\textnormal{approx}}(z;\mathcal{D}+\{z',z''\}) = \ell\big(z,\mathrm{Update}(\hat{f}_{\mathcal D},\{z',z''\})\big)$. As compared to the approximate method, in the tournament-corrected method~\eqref{eq:tournament_corrected_methods} both the training point $Z_i$ and the hypothesized test point $Z_{n+1}^y$ are added into the trained model via a one-step update.
    
    \item \textbf{Example 3: Bayesian posterior predictive density.} Let $z=(x,y)$, $z'=(x',y')$, and $z''=(x'',y'')$. We then define $s_{\mathrm{approx}}(z;\mathcal{D}+\{z',z''\}) = -\sum_{k=1}^K f_{\theta_k}(y\mid x)\cdot f_{\theta_k}(y'\mid x')f_{\theta_k}(y''\mid x'')$, where $\theta_1,\dots,\theta_K\sim\pi(\cdot\mid \mathcal{D})$. Thus, in the tournament-corrected method~\eqref{eq:tournament_corrected_methods}, both the training point $Z_i$ and the hypothesized test point $Z^y_{n+1}$ are added into the posterior with the add-one-in weight-based approximation.\footnote{See Appendix~\ref{app:bayesian_computation_shortcut} for a more formal view of the validity of this randomized score function.}
\end{itemize}

\subsection{Theoretical guarantee}
In practice, we expect that the tournament-corrected method~\eqref{eq:tournament_corrected_methods} will result in approximately $1-\alpha$ coverage, since it provides an approximation to the full conformal prediction set. Theoretically, however, we would need assumptions to ensure that the approximation is accurate. Nonetheless, the tournament correction ensures that, in the worst case, the miscoverage can increase by at most a factor of $2$.
\begin{theorem}[Validity]\label{thm:validity}
    Let $\mathcal{D}_{n+1} = \{Z_1,\dots,Z_{n+1}\}$ be an exchangeable dataset where $Z_i = (X_i,Y_i)$. Let  $s_{\mathrm{approx}}(z;\mathcal D +\{z',z''\})$ be any score that is symmetric in $\mathcal D$ and symmetric in $\{z',z''\}$. Then the conformal prediction set defined in~\eqref{eq:tournament_corrected_methods} 
satisfies the marginal coverage guarantee 
\[\mathbb{P}(Y_{n+1} \in \widehat C^{\textnormal{tourn}}_{n,\alpha}(X_{n+1})) \geq 1 - 2\alpha.\]
\end{theorem}

For the special case of leave-one-out cross-conformal (which is equivalent to our tournament-corrected method, in the setting of Example 0, as described above), this coverage result is established by \citet{barber2021predictive}. 

The broader result of Theorem~\ref{thm:validity} is proved via a generalization of the tournament matrix style argument used in \citet{barber2021predictive} (the complete proof is provided in Appendix~\ref{app:proof_validity}). The tournament matrix arises in the following way: for each training point $i\in[n]$, the test point $n+1$ plays a `game' against this training point, where both $Z_i$ and $Z_{n+1}^y$ are removed from the training set (i.e., we compute scores of the form $s_{\mathrm{approx}}(\,\cdot\,;\mathcal{D}_{n\setminus i} + \{Z_i,Z_{n+1}^y\})$), and then the winner of the `game' is determined by which point, $Z_i$ or $Z_{n+1}^y$, has the higher score. The prediction set $\widehat{C}^{\textnormal{tourn}}_{n,\alpha}(X_{n+1})$ is then constructed by examining how many `games' the test point wins, within this tournament, which is why we refer to this methodology as a \emph{tournament correction}.

\paragraph{A tradeoff: validity versus computation.} The result of Theorem~\ref{thm:validity} shows that the tournament correction allows us to restore theoretical validity when computing an approximation to the full conformal prediction set. However, we remark that this comes at a computational cost, since the tournament correction requires training models with each data point $i\in[n]$ deleted in turn. In some cases it may be possible to avoid additional cost through computational shortcuts (see Appendix~\ref{app:bayesian_computation_shortcut} for one such shortcut in the setting of the Bayesian PPD, via using rejection sampling). In general, we can view this as a tradeoff between statistical validity and computational cost.

\section{Experiments: coverage and length}\label{sec:simulations_coverage}

In this section, we study the empirical behavior of the approximate conformal methods from Section~\ref{sec:background} and their tournament-corrected counterparts from Section~\ref{sec:tourn_correction}. Our goals are twofold:
\begin{itemize}
    \item In practice, existing approximations often work well, providing (nearly) $1-\alpha$ coverage empirically. In such scenarios, we aim to show that the tournament correction leads to minimal changes in the existing methods: the resulting prediction sets are nearly the same.
    \item On the other hand, the tournament correction ensures that coverage can never decrease below $1-2\alpha$, even in a worst-case scenario; in contrast, informal approximations to full conformal prediction do not offer any comparable distribution-free guarantees. Therefore, we aim to show that, in certain challenging settings, existing approximate methods can substantially lose coverage, while the corresponding tournament-corrected methods remain reliable.
\end{itemize}
Taken together, these experiments show that the proposed correction is broadly useful: it does not lead to overly conservative results in scenarios where existing approximations are already accurate, but in regimes where existing approximations break down, the corrected method continues to provide valid uncertainty quantification.\footnote{Code for reproducing all experiment results and tables in the paper is available at \url{https://github.com/aabeshb/Approximating_full_conformal}.}

\subsection{Simulations}\label{sec:simulations}
We first compare each approximation of full conformal (Examples 0, 1, 2, and 3) to its tournament-corrected version in a simulated setting, where data is generated from a Gaussian linear model: the training and test data points are drawn as
\[
X_i \overset{\mathrm{i.i.d.}}{\sim} N(0,I_p), 
\qquad 
Y_i = X_i^\top \beta^\star + \varepsilon_i,
\qquad 
\varepsilon_i \overset{\mathrm{i.i.d.}}{\sim} N(0,1).
\]
The coefficient vector $\beta^\star$ is drawn uniformly from the sphere satisfying $\|\beta^\star\|_2=\sqrt{10}$. Details on the implementation of each of the  methods is provided in Section~\ref{app:regression_simulation_details}.

For each method, we report empirical coverage and average prediction-set length, over $100$ independent trials.

\begin{figure}[ht!]\label{fig:corrected_vs_uncorrected}
    \centering
    \includegraphics[width=\textwidth]{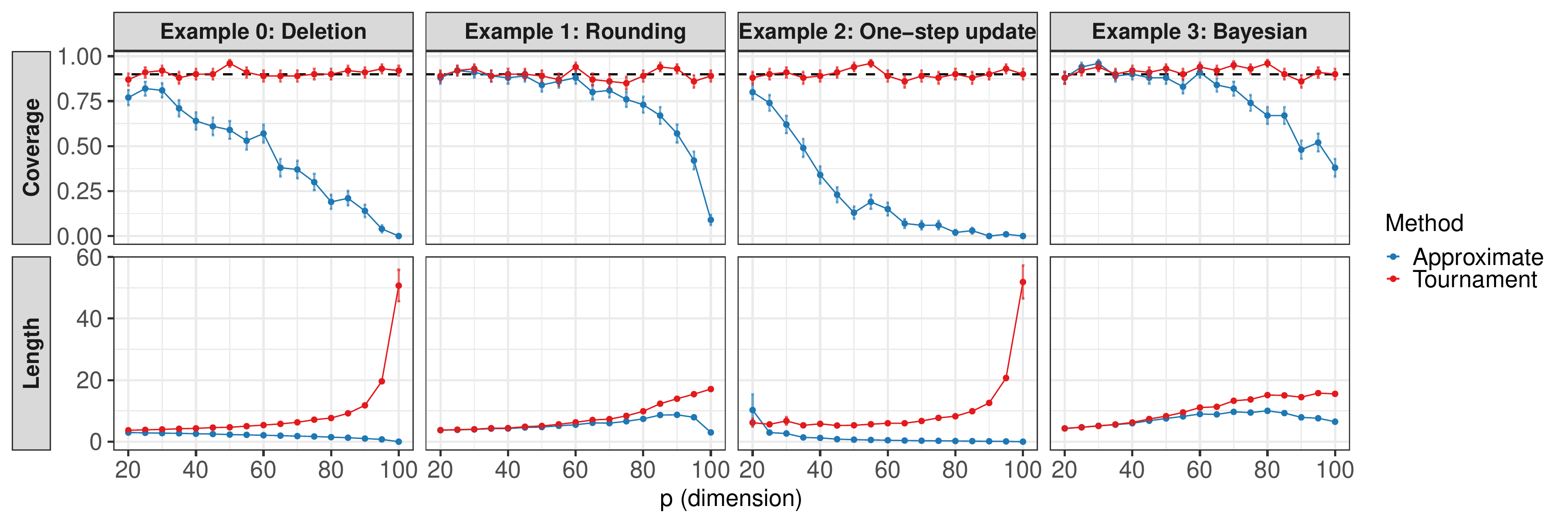}
    \caption{
    Empirical coverage and average prediction-set length for approximate and tournament-corrected conformal methods for $n=100$, with dimension varying over $p=20,25,\dots,100$. A dashed line indicates the nominal coverage level $1-\alpha = 0.9$. Results are averaged over 100 independent trials, with standard error bars shown.}
    \label{fig:linear_simulations}
\end{figure}
Figure~\ref{fig:corrected_vs_uncorrected} show that for smaller values of $p$, both the approximate and the corrected methods have similar performance in terms of both length and coverage. As $p\to n$ and we move towards more unstable regimes, all the approximate methods perform very poorly with coverage falling down far below the target threshold. These show that as the underlying model becomes unstable, the uncorrected  approximations are no longer reliable, whereas all the tournament-corrected methods maintain a level of coverage close to $1-\alpha = 0.9$.

\subsection{Real-data examples}

As we have seen in the simulations, in settings where the existing approximate methods exhibit the desired coverage level empirically, the tournament-corrected method behaves nearly identically: that is, the tournament-corrected method does not result in an overly conservative prediction set in scenarios where no correction is needed. In this next experiment, we study the same question in the setting of real data. Concretely, we replicate two experiments where approximate full conformal methods have been used in practice, and observe that the tournament-corrected method returns nearly identical results in both cases. Both experiments use the Diabetes dataset \citep{efron2004least}. 

First, we consider the rounding approximation (Example 1). For implementing full conformal with rounding (i.e., the uncorrected approximation), we follow the implementation of \citet{lee2025leave} (which studies full conformal as a baseline comparison for a different proposed method), using the residual score~\eqref{eqn:resid_score} for a model trained via ridge-regularized Huber regression. 
Second, we consider the Bayesian PPD add-one-in approximation (Example 3). We replicate an experiment conducted by \citet{fong2021conformal}, using a Gaussian linear regression likelihood with a prior over the parameters. 

For both real data examples, we implement the approximate conformal prediction method exactly as in the respective papers (using the same choice of coverage level $1-\alpha$ as used in those papers), and then implement the tournament-corrected version of the same method. Results are shown in Table~\ref{table:real_data_comparison}. All implementation details are provided in Appendix~\ref{app:realdata_details}. We see that, in both cases, the uncorrected approximation shows approximately the desired level of coverage. Moreover, in both cases, the resulting coverage and prediction set length is nearly identical for the approximate method~\eqref{eq:approximate_conformal_sets} and the tournament-corrected method~\eqref{eq:tournament_corrected_methods}, confirming that in practice, the tournament correction typically is not overly conservative when existing approximations are already providing the desired coverage level empirically.

\begin{table}[ht!]
  \centering
  \caption{Comparison of the length and coverage of approximate and tournament-corrected versions on the \texttt{diabetes} data. Standard errors are shown in parentheses.}
  \label{table:real_data_comparison}

  \setlength{\tabcolsep}{8pt}
  \renewcommand{\arraystretch}{1.12}

  \begin{tabular}{@{} l c l c c @{}}
    \toprule
    \makecell[l]{Method\\[-1pt]\footnotesize Based on paper}
      & $1-\alpha$
      & Type
      & Coverage
      & Length \\
    \midrule

    \multirow{2}{*}{\makecell[l]{Rounding\\[-1pt]\footnotesize \citep{lee2025leave}}}
      & \multirow{2}{*}{0.9}
      & Approximate
      & 0.8879 (0.0037)
      & 2.8244 (0.0045) \\
    &
      & Tournament-corrected
      & 0.9008 (0.0036)
      & 2.8862 (0.0046) \\

    \addlinespace[2pt]

    \multirow{2}{*}{\makecell[l]{Bayesian PPD\\[-1pt]\footnotesize \citep{fong2021conformal}}}
      & \multirow{2}{*}{0.8}
      & Approximate
      & 0.7986 (0.0054)
      & 1.8552 (0.0085) \\
    &
      & Tournament-corrected
      & 0.7986 (0.0055)
      & 1.8552 (0.0084) \\

    \bottomrule
  \end{tabular}
\end{table}

\section{Coverage under stability}\label{sec:coverage_under_stability}
The guarantee of Theorem~\ref{thm:validity} ensures that coverage cannot be worse than $1-2\alpha$, but in practice we expect to see coverage $\approx 1-\alpha$, since the tournament-corrected method approximates the full conformal prediction set. In this section, we show that this can be proved with an additional assumption: if the score function resulting from fitted model is \emph{stable} enough in the sense that it is not too sensitive to small changes in the way the data is used to train the model, then a $2\epsilon$-inflated prediction set which we define below satisfies a stronger coverage guarantee.

\begin{assumption}[stability]\label{ass:score_stability}
There exist \(\epsilon\ge 0\) and \(\nu\in[0,1]\) such that for every $i\in [n]$,
\[
\mathbb P\left(
\left|
s_{\mathrm{approx}}\bigl(Z_{n+1};\mathcal D_{n\setminus i} + \{Z_i,Z_{n+1}\}\bigr)
-
s_{\mathrm{approx}}\bigl(Z_{n+1};\mathcal D_n + \{Z_{n+1}\}\bigr)
\right|
\le \epsilon
\right)
\ge 1-\nu.
\]
\end{assumption}
The probability above is taken with respect to the distribution of $Z_1,\dots,Z_{n+1}$. 
Essentially, the assumption requires that changing the role of a single training point $Z_i$ (i.e., whether it is included in the training data $\mathcal D_n$, or is instead added to $\mathcal D_{n\setminus i}$ via an approximation) is not likely to lead to a large change in the score
at \(Z_{n+1}\).
Assumptions of this flavor are common in learning theory \citep{bousquet2002stability,hardt2016train}, and have also been used recently to justify computationally efficient approximations to full conformal prediction
\citep{barber2021predictive,ndiaye2022stable,lee2025leave}.

Next, define an inflated version of the tournament-corrected set:
\begin{multline}\label{eq:inflated_interval}
\widehat C_{n,\alpha}^{\mathrm{tourn},2\epsilon}(X_{n+1})
:=
\Bigg\{
y:\;
\sum_{i=1}^n
\mathbf 1\big\{
s_{\mathrm{approx}}\bigl(Z_{n+1}^y;\mathcal D_{n\setminus i } + \{Z_{n+1}^y,Z_i\}\bigr) \\
>
s_{\mathrm{approx}}\bigl(Z_i;\mathcal D_{n\setminus i } + \{Z_i,Z_{n+1}^y\}\bigr)+2\epsilon
\big\}
<
(1-\alpha)(n+1)
\Bigg\}.
\end{multline}

\begin{theorem}[Coverage under stability]\label{thm:validity_under_stability}

Under the setting and notation of Theorem~\ref{thm:validity}, if also Assumption~\ref{ass:score_stability} holds, then the $2\epsilon$-inflated set as defined in~\eqref{eq:inflated_interval} satisfies the coverage guarantee
\[
\mathbb P\left(
Y_{n+1}\in \widehat C_{n,\alpha}^{\mathrm{tourn},2\epsilon}(X_{n+1})
\right)
\ge 1-\alpha-3\sqrt{\nu}.
\]
\end{theorem}

When \(\nu\) is small, the bound \(1-\alpha-3\sqrt{\nu}\) is therefore close to the nominal
\(1-\alpha\) level, providing a stronger coverage guarantee as compared to the \(1-2\alpha\) guarantee of
Theorem~\ref{thm:validity}. The proof of this theorem is given in Appendix~\ref{app:proof_validity_under_stability}.

\subsection{Stability comparison}

Theorem~\ref{thm:validity_under_stability} shows that the tournament-corrected methods admit a stronger coverage guarantee under a single stability condition, namely Assumption~\ref{ass:score_stability}. A natural question is whether the original approximate prediction sets in~\eqref{eq:approximate_conformal_sets} can also enjoy a comparable near-\((1-\alpha)\) guarantee: in other words, if we are willing to assume stability, is it still necessary to apply the tournament correction?

In fact, we will now see that the tournament correction continues to be useful. This is because the type of stability assumption needed for the (uncorrected) approximation~\eqref{eq:approximate_conformal_sets} to guarantee coverage, is much stronger than the type of stability required by Assumption~\ref{ass:score_stability}. Specifically, in order for $\widehat{C}_{n,\alpha}^{\mathrm{approx}}(X_{n+1})$ to lead to a similar coverage guarantee as Theorem~\ref{thm:validity_under_stability}, we would need to require
that
\begin{equation}\label{eq:stab_stronger}\begin{split}&
\mathbb P\left(
\left|
s_{\mathrm{approx}}(Z_{n+1};\mathcal D_n+\{Z_{n+1}\})
-
s(Z_{n+1};\mathcal D_{n+1})
\right|
\le \epsilon
\right)
\ge 1-\nu,\\
&
\mathbb P\left(
\left|
s_{\mathrm{approx}}(Z_i;\mathcal D_n+\{Z_{n+1}\})
-
s(Z_i;\mathcal D_{n+1})
\right|
\le \epsilon
\right)
\ge 1-\nu.
\end{split}\end{equation}
These conditions would lead to a coverage result for an inflated prediction set, analogous to the result of Theorem~\ref{thm:validity_under_stability} (see Appendix~\ref{app:validity_approx_methods} for details). Results of this type appear in the literature for various settings \citep{barber2021predictive,ndiaye2022stable,lee2025leave}.

However, the first condition in~\eqref{eq:stab_stronger} is substantially more restrictive than Assumption~\ref{ass:score_stability}: it requires that training only approximately on $Z_{n+1}$, will not lead to a large change in the score for $Z_{n+1}$, i.e., \emph{for the same data point that was removed}. For example, if we use rounding (so $Y_{n+1}$ is replaced by a value on the grid), an overparameterized method is likely to substantially change its prediction at $X_{n+1}$---and thus this type of stability condition would not hold. (See \citet{barber2021predictive} for further discussion of the difference in these stability conditions for the specific setting of Example 0, i.e., where the approximation is carried out via data point deletion).

To illustrate this difference empirically, for each of the four different methods, we estimate the stability properties of the score function. Specifically, for each method, we plot the values $(\nu,\epsilon)$ for which the conditions~\eqref{eq:stab_stronger} hold, in order to assess whether the uncorrected approximation would likely have approximate coverage; we also plot the values of $(\nu,\epsilon)$ for which Assumption~\ref{ass:score_stability} holds, in order to assess whether the tournament-corrected approximation would have $\approx 1-\alpha$ coverage as in Theorem~\ref{thm:validity_under_stability}. The data is drawn from the same distribution as in Section~\ref{sec:simulations} and the fitted model and all the implementational details are the same as Section~\ref{sec:simulations}. For this simulation, we consider $n=100,p=20$ which is a relatively stable regime and report the values over $200$ independent trials.

Even then across all four methods, the tournament-corrected methods exhibit stability curves that lie substantially below the curves corresponding to the uncorrected approximations. This indicates that Assumption~\ref{ass:score_stability} holds for small values of $\epsilon$ and $\nu$, leading to favorable coverage properties (via Theorem~\ref{thm:validity_under_stability}) for the tournament-corrected method; in contrast, the first stability condition in~\eqref{eq:stab_stronger} only holds at much larger values of $\epsilon$ and $\nu$, meaning that the uncorrected approximations may not lead to coverage.

\begin{figure}[ht!]
    \centering
    \includegraphics[width = \linewidth]{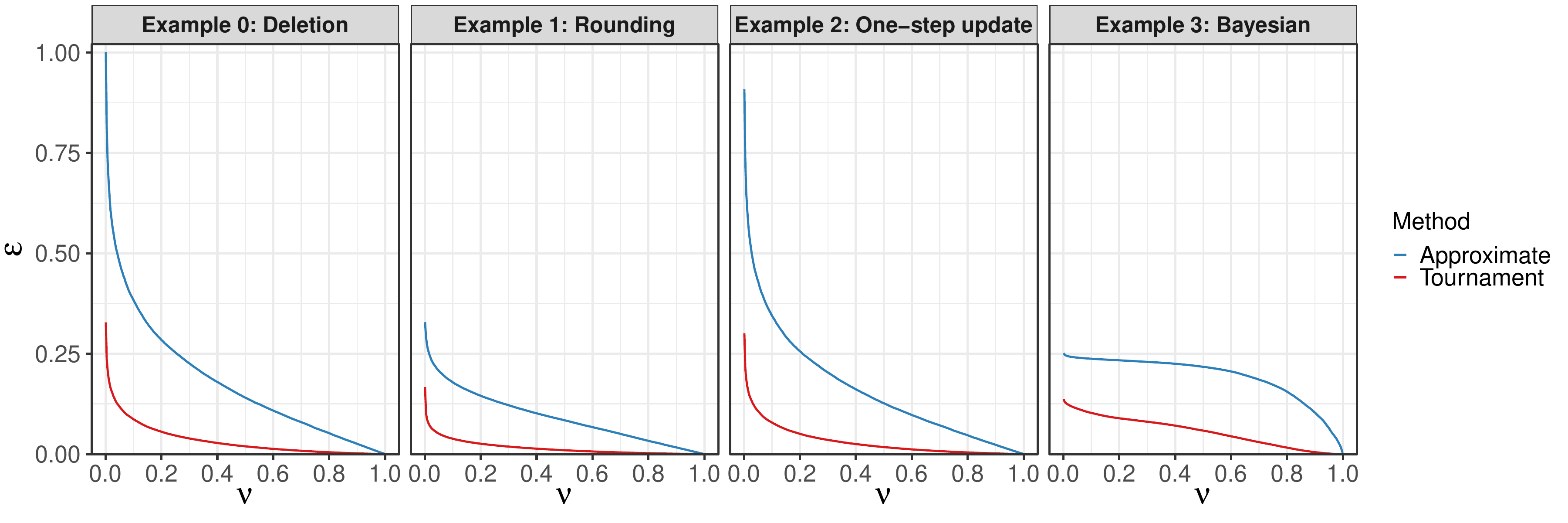}
    \caption{Comparison of the stability properties for the uncorrected approximation (see~\eqref{eq:stab_stronger}) versus the tournament-corrected approximation (see Assumption~\ref{ass:score_stability}).}
    \label{fig:stability_corrected_vs_uncorrected}
\end{figure}

\section{Discussion}\label{sec:discussion}
In this work, we develop a principled way of approximating full conformal prediction with provable distribution-free coverage guarantees. While there exist quite a few methods that are used as an approximation of full conformal, such as using a grid of $\mathcal{Y}$ values or adding in the test point via a one-step model update instead of refitting the model for each value of $y$, in general such approximations do not offer assumption-free theoretical coverage guarantees, and can fail in unstable regimes. Our proposed method offers a \emph{tournament correction} of these approximations, which ensures at least $1-2\alpha$ coverage with no additional assumptions. 

\emph{Limitations and future work.} It is important to note that implementing the tournament correction will typically increase computational cost: namely, $n$ times the cost incurred by their approximate counterparts. In Appendix~\ref{app:bayesian_computation_shortcut}, we discuss how we can get around this issue for the Bayesian case by using rejection sampling techniques, however, in the general case this may not always be possible. \ To address this, we may also define $K$-fold versions of the tournament correction (analogous to $K$-fold cross-conformal prediction, for the case of data deletion, i.e., Example 0). 
We leave the question of such extensions to future work.

\subsection*{Acknowledgements}
A.B. and R.F.B. were partially supported by the National Science Foundation via grant DMS-2023109. R.F.B. was partially supported by the Office of Naval Research via grant N00014-24-1-2544.

\bibliographystyle{plainnat}
\bibliography{ref}


\appendix
\numberwithin{theorem}{section}
\renewcommand{\thetheorem}{\Alph{section}.\arabic{theorem}}

\section{Proofs of theorems}\label{app:proofs_theorems}

\subsection{Proof of Theorem~\ref{thm:validity}}\label{app:proof_validity}

The proof follows the core idea of proof of \citet[Theorem 1]{barber2021predictive}. We will construct a $(n+1)\times(n+1)$ matrix $A$, as a function of the data, such that $A$ is a \emph{tournament matrix}, meaning that $A_{ij}\in\{0,1\}$, and $A_{ij}+A_{ji}\leq 1$, for all $i,j$ (that is, $A_{ij}$ can be viewed as an indicator of whether team $i$ wins its game against team $j$---and thus we cannot have both $A_{ij}=1$ and $A_{ji}=1$). We will establish that the coverage event can be determined by the $(n+1)$st row of $A$, with
\begin{equation}\label{eqn:tourn_step1}
    Y_{n+1}\not\in\widehat{C}^{\textnormal{tourn}}_{n,\alpha}(X_{n+1}) \ \Longleftrightarrow \ \sum_{i=1}^{n+1}A_{n+1,i} \geq (1-\alpha)(n+1),
\end{equation}
and that $A$ satisfies a row/column-wise exchangeability property,
\begin{equation}\label{eqn:tourn_step2}
A \stackrel{\textnormal{d}}{=} \Pi A \Pi^\top \textnormal{ for any permutation matrix $\Pi\in\{0,1\}^{(n+1)\times(n+1)}$.}
\end{equation}
As in the proof of \citet[Theorem 1]{barber2021predictive}, these properties are sufficient to prove that
\[\mathbb P(Y_{n+1}\in\widehat{C}^{\textnormal{tourn}}_{n,\alpha}(X_{n+1})) \geq 1-2\alpha,\]
as desired: to summarize the proof idea, this is because the bound
\[\sum_{j=1}^{n+1}\mathbf 1\left\{ \sum_{i=1}^{n+1} A_{j,i} \geq (1-\alpha)(n+1)\right\} \leq 2\alpha(n+1)\]
must hold \emph{deterministically} for any tournament matrix (i.e., it is impossible for more than $2\alpha(n+1)$ of the teams to win $\geq (1-\alpha)(n+1)$ of their games, in any tournament), while the exchangeability property~\eqref{eqn:tourn_step2} ensures that 
\[\mathbb P\left(\sum_{i=1}^{n+1} A_{n+1,i} \geq (1-\alpha)(n+1)\right) = \mathbb P\left(\sum_{i=1}^{n+1} A_{j,i} \geq (1-\alpha)(n+1)\right)\textnormal{ for all $j$}.\]

Now we construct the tournament matrix $A$, and verify that~\eqref{eqn:tourn_step1} and~\eqref{eqn:tourn_step2} both hold. First we define a matrix of scores, $S\in\mathbb R^{(n+1)\times(n+1)}$, with
\[
    \begin{cases}
        S_{ij} := s_{\mathrm{approx}}(Z_i;\mathcal{D}_{[n+1]\setminus\{i,j\}} + \{Z_i,Z_j\}), & i \neq j,\\
        S_{ii} := +\infty, & i=j.
    \end{cases}
    \]
By definition of the tournament-corrected prediction set~\eqref{eq:tournament_corrected_methods}, 
\[Y_{n+1}\not\in\widehat{C}^{\textnormal{tourn}}_{n,\alpha}(X_{n+1}) \ \Longleftrightarrow \ \sum_{i=1}^n \mathbf 1\{S_{n+1,i}>S_{i,n+1}\}\geq  (1-\alpha)(n+1).\]
Next define the matrix $A=A(S)$ with entries
\[A_{ij} = \mathbf 1\{S_{ij}> S_{ji}\}.\]
By construction, $A$ is a tournament matrix, and satisfies~\eqref{eqn:tourn_step1}.
Moreover, by exchangeability of the data together with the fact that $s_{\mathrm{approx}}(z;\mathcal D+\{z',z''\})$ is assumed to be symmetric in $\mathcal D$, we have
\[S \stackrel{\textnormal{d}}{=} \Pi S \Pi^\top\]
for any permutation matrix $\Pi$. This row/column-wise exchangeability property is then inherited by $A=A(S)$, since
\[A(S) \stackrel{\textnormal{d}}{=} A(\Pi S\Pi^\top) = \Pi A(S) \Pi^\top.\]
This verifies~\eqref{eqn:tourn_step2}, and thus completes the proof.

\subsubsection{Extension to randomized scores}\label{app:extension_random_scores}
We next extend the result of Theorem~\ref{thm:validity} to the setting of a randomized score function, to accommodate examples such as Bayesian PPD (Example 3). 

Let $\xi_{n+1,i}\in[0,1]$ and $\xi_{i,n+1}\in[0,1]$ denote random variables (i.e., seeds used for randomization), drawn independently of the data. We now define the tournament-corrected prediction set as
\begin{equation}\begin{split}
     \widehat C_{n,\alpha}^{\mathrm{tourn}}(X_{n+1}) &:= \Bigg\{y:\sum_{i=1}^n\mathbf{1}\Big\{ s_{\mathrm{approx}}(Z_{n+1}^y;\mathcal{D}_{n\setminus i} + \{Z_i,Z_{n+1}^y\};\xi_{n+1,i})  \\ 
     & > s_{\mathrm{approx}}(Z_i;\mathcal{D}_{n\setminus i} + \{Z_i,Z_{n+1}^y\};\xi_{i,n+1})\Big\} <(1-\alpha)(n+1)\Bigg\},
\end{split}\label{eq:tournament_corrected_methods_randomized}\end{equation}
where now both the score $s_{\mathrm{approx}}(Z_{n+1}^y;\mathcal{D}_{n\setminus i} + \{Z_i,Z_{n+1}^y\};\xi_{n+1,i})$ for the test point, and the score $s_{\mathrm{approx}}(Z_i;\mathcal{D}_{n\setminus i} + \{Z_i,Z_{n+1}^y\};\xi_{i,n+1})$ for the $i$th training point, use the additional argument of a random seed (which permits randomization, e.g., a score constructed via random samples drawn from a posterior, as for Bayesian PPD). 

We will need to assume that the random seeds $\xi_{n+1,i}$ and $\xi_{i,n+1}$ can be taken to be entries of a matrix $\xi\in[0,1]^{(n+1)\times(n+1)}$, satisfying row/column-wise exchangeability:
\begin{equation}\label{eqn:random_seed}
    \xi \stackrel{\textnormal{d}}{=}\Pi \xi\Pi^\top\textnormal{ for any permutation matrix $\Pi\in\{0,1\}^{(n+1)\times(n+1)}$.}
\end{equation}
For instance, this holds if the $2n$ same random seeds $\xi_{n+1,1},\dots,\xi_{n+1,n},\xi_{1,n+1},\dots,\xi_{n,n+1}$ are all equal (e.g., all equal to the same random variable $\xi\sim\textnormal{Unif}[0,1]$), or alternatively, are i.i.d.\ draws from $\textnormal{Unif}[0,1]$.
\begin{theorem}\label{thm:validity_randomized}
    In the setting of Theorem~\ref{thm:validity}, now consider the randomized construction as in~\eqref{eq:tournament_corrected_methods_randomized}. Assume also that the random seeds satisfy~\eqref{eqn:random_seed} and are independent of the data. Then
\[\mathbb{P}(Y_{n+1} \in \widehat C(X_{n+1})) \geq 1 - 2\alpha.\]
\end{theorem}
\begin{proof}
    The proof is essentially identical to that of Theorem~\ref{thm:validity}. We define a score matrix $S$ with entries
    \[
    \begin{cases}
        S_{ij} := s_{\mathrm{approx}}(Z_i;\mathcal{D}_{[n+1]\setminus\{i,j\}} + \{Z_i,Z_j\}; \xi_{ij}), & i \neq j,\\
        S_{ii} := +\infty, & i=j,
    \end{cases}
    \]
    and observe that $S$ satisfies row/column-wise exchangeability by assumption on the data and the random seeds. The rest of the proof proceeds exactly as before.
\end{proof}

\subsection{Proof of Theorem~\ref{thm:validity_under_stability}} \label{app:proof_validity_under_stability}
This proof follows a similar argument as \citet[Theorem 5]{barber2021predictive}, which proves a related result for the setting of leave-one-out cross-conformal or cross-validation methods.

For \(i=1,\dots,n\), define
\[
R_i^y
:=
s_{\mathrm{approx}}\bigl(Z_i;\mathcal D_{[n+1]\setminus \{i\}}^y+\{Z_i\}\bigr),
\]
and define
\[
R_{n+1}^y
:=
s_{\mathrm{approx}}\bigl(Z_{n+1}^y;\mathcal D_n+\{Z_{n+1}^y\}\bigr).
\]
Thus, when \(y=Y_{n+1}\), writing \(R_i:=R_i^{Y_{n+1}}\), we have
\[
R_i=s_{\mathrm{approx}}\bigl(Z_i;\mathcal D_{[n+1]\setminus \{i\}}+\{Z_i\}\bigr),
\qquad i=1,\dots,n+1.
\]

Since \(Z_1,\dots,Z_{n+1}\) are exchangeable and
\(s_{\mathrm{approx}}(z;\mathcal D+\{z\})\) is symmetric in \(\mathcal D\),
the scores \(R_1,\dots,R_{n+1}\) are exchangeable. Hence the oracle conformal set
\[
\widetilde C_{n,\alpha_0}(X_{n+1})
:=
\left\{
y:\;
\sum_{i=1}^n
\mathbf 1\{R_{n+1}^y>R_i^y\}
<
(1-\alpha_0)(n+1)
\right\}
\]
has coverage at least \(1-\alpha_0\), i.e.,
\begin{equation}\label{eq:oracle_cov_new}
\mathbb P\left(Y_{n+1}\not\in \widetilde C_{n,\alpha_0}(X_{n+1})\right) = \mathbb P\left(
\sum_{i=1}^n \mathbf 1\{R_{n+1}>R_i\}\ge (1-\alpha_0)(n+1)
\right)
\le \alpha_0.
\end{equation}

Now fix \(i\in[n]\), and define
\[
\Delta_i
:=
\left|
s_{\mathrm{approx}}\bigl(Z_{n+1};\mathcal D_{n\setminus i}+\{Z_{n+1},Z_i\}\bigr)
-
R_{n+1}
\right|.
\]
By Assumption~\ref{ass:score_stability},
\[
\mathbb P(\Delta_i\le \epsilon)\ge 1-\nu.
\]

Next define
\[
\Delta_i'
:=
\left|
s_{\mathrm{approx}}\bigl(Z_i;\mathcal D_{n\setminus i}+\{Z_i,Z_{n+1}\}\bigr)
-
R_i
\right|.
\]
By Assumption~\ref{ass:score_stability} (together with the assumption that the data is exchangeable),
\[
\mathbb P(\Delta_i'\le \epsilon)\ge 1-\nu.
\]
If we let
\[
G_i:=\{\Delta_i\le \epsilon\}\cap\{\Delta_i'\le \epsilon\},
\]
then by the union bound,
\begin{equation}\label{eq:Gi_bound_new}
\mathbb P(G_i^c)\le 2\nu.
\end{equation}

Now suppose that \(G_i\) holds and \(R_{n+1}\le R_i\). Then
\begin{align*}
s_{\mathrm{approx}}\bigl(Z_{n+1};\mathcal D_{n\setminus i}+\{Z_{n+1},Z_i\}\bigr)
&\le
R_{n+1}+\epsilon \\
&\le
R_i+\epsilon \\
&\le
s_{\mathrm{approx}}\bigl(Z_i;\mathcal D_{n\setminus i}+\{Z_i,Z_{n+1}\}\bigr)+2\epsilon.
\end{align*}
Hence
\begin{multline*}
\mathbf 1\left\{
s_{\mathrm{approx}}\bigl(Z_{n+1};\mathcal D_{n\setminus i}+\{Z_{n+1},Z_i\}\bigr)
>
s_{\mathrm{approx}}\bigl(Z_i;\mathcal D_{n\setminus i}+\{Z_i,Z_{n+1}\}\bigr)+2\epsilon
\right\}\\
\le
\mathbf 1\{R_{n+1}>R_i\}+\mathbf 1\{G_i^c\}.
\end{multline*}
Summing over \(i=1,\dots,n\),
\begin{align}
&\sum_{i=1}^n
\mathbf 1\left\{
s_{\mathrm{approx}}\bigl(Z_{n+1};\mathcal D_{n\setminus i}+\{Z_{n+1},Z_i\}\bigr)
>
s_{\mathrm{approx}}\bigl(Z_i;\mathcal D_{n\setminus i}+\{Z_i,Z_{n+1}\}\bigr)+2\epsilon
\right\}
\notag\\
&\qquad\le
\sum_{i=1}^n \mathbf 1\{R_{n+1}>R_i\}
+
\sum_{i=1}^n \mathbf 1\{G_i^c\}.
\label{eq:compare_counts_new}
\end{align}

Set \(\alpha_0:=\alpha+2\sqrt{\nu}\). Using \eqref{eq:compare_counts_new},
\begin{align*}
&\mathbb P\Bigg(
\sum_{i=1}^n
\mathbf 1\bigg\{
s_{\mathrm{approx}}\bigl(Z_{n+1};\mathcal D_{n\setminus i}+\{Z_{n+1},Z_i\}\bigr)\\&\hspace{1in}
>
s_{\mathrm{approx}}\bigl(Z_i;\mathcal D_{n\setminus i}+\{Z_i,Z_{n+1}\}\bigr)+2\epsilon
\bigg\}
\ge (1-\alpha)(n+1)
\Bigg)
\\
&\qquad\le
\mathbb P\left(
\sum_{i=1}^n \mathbf 1\{R_{n+1}>R_i\}\ge (1-\alpha_0)(n+1)
\right)
+
\mathbb P\left(
\sum_{i=1}^n \mathbf 1\{G_i^c\}\ge 2\sqrt{\nu}(n+1)
\right).
\end{align*}
The first term is at most \(\alpha_0\) by \eqref{eq:oracle_cov_new}. For the second term, by
\eqref{eq:Gi_bound_new} and Markov's inequality,
\[
\mathbb P\left(
\sum_{i=1}^n \mathbf 1\{G_i^c\}\ge 2\sqrt{\nu}(n+1)
\right)
\le
\frac{\sum_{i=1}^n \mathbb P(G_i^c)}{2\sqrt{\nu}(n+1)}
\le
\frac{2n\nu}{2\sqrt{\nu}(n+1)}
\le
\sqrt{\nu}.
\]
Therefore,
\[
\mathbb P\left(
Y_{n+1}\notin \widehat C_{n,\alpha}^{\mathrm{infl},2\epsilon}(X_{n+1})
\right)
\le
\alpha+2\sqrt{\nu}+\sqrt{\nu}
=
\alpha+3\sqrt{\nu},
\]
which completes the proof.

\section{Simulation details of Section~\ref{sec:simulations}}
\label{app:regression_simulation_details}

We compare all the four approximation schemes for full conformal prediction:
deletion, rounding, one-step update, and Bayesian add-one-in posterior
predictive approximation. For each scheme, we evaluate both the original
approximate method as described in Section~\ref{sec:background} and the
corresponding tournament-corrected method proposed in this paper in
Section~\ref{sec:tourn_correction}. For the non-Bayesian methods, the fitted
model is ordinary least squares without an intercept. When the design matrix is
rank-deficient or underdetermined, we use the Moore--Penrose pseudoinverse and the
score function is the absolute residual.

In all simulations, we use sample size $n=100$ nominal coverage level $1-\alpha = 0.9$,
and vary the dimension over $p\in\{20,25,30,\ldots,100\}$.
For each value of \(p\), we run \(100\) independent trials. 

For each method, empirical coverage is computed as
\[
\widehat{\mathrm{Cov}}
=
\frac{1}{B}
\sum_{b=1}^B
\mathbf 1\left\{
Y_{n+1}^{(b)}
\in
\widehat C^{(b)}(X_{n+1}^{(b)})
\right\},
\]
averaged over $B=100$ trials.
The average length is the average Lebesgue measure of the returned prediction set,
\[
\widehat{\mathrm{Length}}
=
\frac{1}{B}
\sum_{b=1}^B \textnormal{Leb}\big(
\widehat C^{(b)}(X_{n+1}^{(b)})
\big).
\]

\subsection{Example 0: Deletion}

\paragraph{Approximate.}For the deletion approximation, we use
\[
s_{\mathrm{approx}}^{\mathrm{delete}}(z;\mathcal D+\{z'\})
=
s(z;\mathcal D),
\]
where, for \(z=(x,y)\),
\[
s(z;\mathcal D)
=
|y-x^\top\widehat\beta(\mathcal D)|.
\]
Here \(\widehat\beta(\mathcal D)\) is the ordinary least-squares fit on
\(\mathcal D\). Thus the uncorrected deletion method fits OLS once on
\(\mathcal D_n\), computes
\[
R_i
=
s_{\mathrm{approx}}^{\mathrm{delete}}(Z_i;\mathcal D_n+\{Z_{n+1}^y\})
=
|Y_i-X_i^\top\widehat\beta(\mathcal D_n)|,
\qquad i\in[n],
\]
and
\[
s_{\mathrm{approx}}^{\mathrm{delete}}(Z_{n+1}^y;\mathcal D_n+\{Z_{n+1}^y\})
=
|y-X_{n+1}^\top\widehat\beta(\mathcal D_n)|.
\]

\paragraph{Tournament-corrected.} For the tournament-corrected deletion method, we use
\[
s_{\mathrm{approx}}^{\mathrm{delete}}(z;\mathcal D+\{z',z''\})
=
s(z;\mathcal D)
=
|y-x^\top\widehat\beta(\mathcal D)|.
\]
Thus, for each \(i\in[n]\), the pairwise comparison between \(Z_i\) and
\(Z_{n+1}^y\) is computed using the leave-one-out fit
\(\widehat\beta(\mathcal D_{n\setminus i})\).

\subsection{Example 1: Rounding}

\paragraph{Constructing the grid.} For the existing theory to work, the grid must be fixed before hand without looking at the data. This is so that the rounding operation treats the training and the test data symmetrically in the sense that all of them are rounded with respect to the same grid which is chosen independently of the training data. However, \cite{chen2018discretized} show that if the grid is chosen as a function of the range of the training $Y$ values (i.e., $\min_{i=1,\dots,n} Y_i$ and $\max_{i=1,\dots,n} Y_i$), this leads to at most $\tfrac{2}{n+1}$ additional loss of coverage. Hence, for the rounding approximation, in each trial
the grid is constructed from the observed training responses. Let
\[
Y_{\min}=\min_{1\le i\le n}Y_i,
\qquad
Y_{\max}=\max_{1\le i\le n}Y_i,
\qquad
R_Y=Y_{\max}-Y_{\min}.
\]
The grid consists of \(M=10\) equally spaced points $y^{(1)},\dots,y^{(M)}$ between
\[
Y_{\min}-0.02R_Y
\qquad\text{and}\qquad
Y_{\max}+0.02R_Y.
\]
We write \([y]\) for the nearest grid point to \(y\). The real line is
partitioned into the corresponding Voronoi cells, with cell boundaries given by
the midpoints between adjacent grid values.

\paragraph{Approximate.}For the uncorrected rounding method, the approximate score is
\[
s_{\mathrm{approx}}^{\mathrm{round}}(z;\mathcal D+\{z'\})
=
s\bigl(z;\mathcal D\cup\{(x',[y'])\}\bigr),
\qquad z'=(x',y').
\]
Using the absolute residual score, this is
\[
s_{\mathrm{approx}}^{\mathrm{round}}(z;\mathcal D+\{z'\})
=
|y-x^\top \widehat\beta(\mathcal D\cup\{(x',[y'])\})|.
\]
Thus, for each grid point \(g_m\), we fit OLS on
\[
\mathcal D_n\cup\{(X_{n+1},g_m)\}.
\]
For candidate values \(y\) whose rounded value is \([y]=y^{(m)}\), the score at the
candidate point is
\[
s_{\mathrm{approx}}^{\mathrm{round}}
(Z_{n+1}^y;\mathcal D_n+\{Z_{n+1}^y\})
=
|y-X_{n+1}^\top\widehat\beta(\mathcal D_n\cup\{(X_{n+1},y^{(m)})\})|,
\]
and the training scores are
\[
s_{\mathrm{approx}}^{\mathrm{round}}
(Z_i;\mathcal D_n+\{Z_{n+1}^y\})
=
|Y_i-X_i^\top\widehat\beta(\mathcal D_n\cup\{(X_{n+1},y^{(m)})\})|.
\]

\paragraph{Tournament-corrected.}For the tournament-corrected rounding method, we use the corrected score
\[
s_{\mathrm{approx}}^{\mathrm{round}}(z;\mathcal D+\{z',z''\})
=
s\bigl(z;\mathcal D\cup\{(x',[y'])\}\cup\{(x'',[y''])\}\bigr),
\]
where \(z=(x,y)\) and \(z'=(x',y')\) and \(z''=(x'',y'')\). With the absolute residual score,
\[
s_{\mathrm{approx}}^{\mathrm{round}}(z;\mathcal D+\{z',z''\})
=
|y-x^\top
\widehat\beta(\mathcal D\cup\{(x',[y'])\}\cup\{(x'',[y''])\})|.
\]
Therefore, for the pairwise comparison between \(Z_i\) and \(Z_{n+1}^y\), the
model is fit on
\[
\mathcal D_{n\setminus i}
\cup
\{(X_i,[Y_i])\}
\cup
\{(X_{n+1},[y])\}.
\]

\subsection{Example 2: One-step update implementation}
\label{app:one-step_details}

\paragraph{Approximate.} In the simulation, the fitted model is
linear, \(f_\theta(x)=x^\top\theta\), and
\[
\widehat\theta(\mathcal D)
=
\widehat\beta(\mathcal D)
=
\arg\min_{\theta\in\mathbb R^p}
\sum_{(X_j,Y_j)\in\mathcal D}
(Y_j-X_j^\top\theta)^2.
\]
The step size is $\eta=10$. For a single added point \(z'=(x',y')\), the one-step update used in the
simulation is
\[
\mathrm{Update}_{\eta}(\widehat f_{\mathcal D},z')
=
f_{\widetilde\theta},
\]
where
\[
\widetilde\theta
=
\widehat\theta(\mathcal D)
+
\frac{\eta}{|\mathcal D|+1}
x'
\left(y'-x'^\top\widehat\theta(\mathcal D)\right).
\]
This is the gradient-descent update for the squared-error loss, evaluated at
the OLS solution on \(\mathcal D\).

The uncorrected one-step score is
\[
s_{\mathrm{approx}}^{\mathrm{one-step}}(z;\mathcal D+\{z'\})
=
\left|
y-
\left[
\mathrm{Update}_{\eta}(\widehat f_{\mathcal D},z')
\right](x)
\right|,
\qquad z=(x,y).
\]
Thus, in the uncorrected method,
\[
s_{\mathrm{approx}}^{\mathrm{one-step}}
(Z_{n+1}^y;\mathcal D_n+\{Z_{n+1}^y\})
=
\left|
y-
\left[
\mathrm{Update}_{\eta}(\widehat f_{\mathcal D_n},Z_{n+1}^y)
\right](X_{n+1})
\right|,
\]
and
\[
s_{\mathrm{approx}}^{\mathrm{one-step}}
(Z_i;\mathcal D_n+\{Z_{n+1}^y\})
=
\left|
Y_i-
\left[
\mathrm{Update}_{\eta}(\widehat f_{\mathcal D_n},Z_{n+1}^y)
\right](X_i)
\right|.
\]
The approximate one-step prediction set is computed exactly as a union of
intervals. To do this, let
\[
\widehat\beta=\widehat\beta(\mathcal D_n),
\qquad
d=X_{n+1}^\top\widehat\beta,
\qquad
e_i=Y_i-X_i^\top\widehat\beta.
\]
For \(t=y-d\), define
\[
\kappa_i=\frac{\eta}{n+1}X_i^\top X_{n+1},
\qquad
\kappa_{n+1}=\frac{\eta}{n+1}\|X_{n+1}\|_2^2.
\]
Then
\[
s_{\mathrm{approx}}^{\mathrm{one-step}}
(Z_{n+1}^y;\mathcal D_n+\{Z_{n+1}^y\})
=
|1-\kappa_{n+1}|\,|t|,
\]
and
\[
s_{\mathrm{approx}}^{\mathrm{one-step}}
(Z_i;\mathcal D_n+\{Z_{n+1}^y\})
=
|e_i-\kappa_i t|.
\]
The only points where the rank of the test score among the training scores can
change are the solutions of
\[
|1-\kappa_{n+1}|\,|t|
=
|e_i-\kappa_i t|,
\qquad i\in[n],
\]
together with \(t=0\). Equivalently, the breakpoints are
\[
y=d,
\qquad
y=d+\frac{e_i}{\kappa_i+|1-\kappa_{n+1}|},
\qquad
y=d+\frac{e_i}{\kappa_i-|1-\kappa_{n+1}|},
\]
whenever the denominators are nonzero. The simulation sorts these breakpoints,
evaluates the approximate conformal acceptance condition on each induced
interval, and returns the union of accepted intervals. Coverage is evaluated by
checking whether \(Y_{n+1}\) lies in this union, and length is the total length
of the union.

\paragraph{Tournament-corrected.} For the tournament-corrected one-step method, the update uses both points in
the pairwise comparison. Namely, for \(z=(x,y)\) and \(z'=(x',y')\), we use
\[
s_{\mathrm{approx}}^{\mathrm{one-step}}(z;\mathcal D+\{z',z''\})
=
\left|
y-
\left[
\mathrm{Update}_{\eta}(\widehat f_{\mathcal D},\{z',z''\})
\right](x)
\right|,
\]
where
\[
\mathrm{Update}_{\eta}(\widehat f_{\mathcal D},\{z',z''\})
=
f_{\widetilde\theta},
\]
and
\[
\widetilde\theta
=
\widehat\theta(\mathcal D)
+
\frac{\eta}{|\mathcal D|+2}
\left[
x'\left(y'-x'{}^\top\widehat\theta(\mathcal D)\right)
+
x''\left(y''-{x''}^\top\widehat\theta(\mathcal D)\right)
\right].
\]
For the comparison between \(Z_i\) and \(Z_{n+1}^y\), we take
\(\mathcal D=\mathcal D_{n\setminus i}\). Let
\[
\widehat\beta_{-i}
=
\widehat\beta(\mathcal D_{n\setminus i}),
\qquad
r_i=Y_i-X_i^\top\widehat\beta_{-i},
\qquad
d_i=X_{n+1}^\top\widehat\beta_{-i}.
\]
Also define
\[
a_i=\frac{\eta}{n+1}X_i^\top X_{n+1},
\qquad
g_i=\frac{\eta}{n+1}\|X_i\|_2^2,
\qquad
b=\frac{\eta}{n+1}\|X_{n+1}\|_2^2.
\]
Writing \(t_i=y-d_i\), the two scores in the pairwise comparison are
\[
s_{\mathrm{approx}}^{\mathrm{one-step}}
(Z_{n+1}^y;\mathcal D_{n\setminus i}+\{Z_i,Z_{n+1}^y\})
=
\left|
(1-b)t_i-a_i r_i
\right|,
\]
and
\[
s_{\mathrm{approx}}^{\mathrm{one-step}}
(Z_i;\mathcal D_{n\setminus i}+\{Z_i,Z_{n+1}^y\})
=
\left|
(1-g_i)r_i-a_i t_i
\right|.
\]
The corrected one-step candidate \(y\) is accepted (i.e., included in the prediction set) if
\[
\sum_{i=1}^n
\mathbf 1\left\{
s_{\mathrm{approx}}^{\mathrm{one-step}}
(Z_{n+1}^y;\mathcal D_{n\setminus i}+\{Z_i,Z_{n+1}^y\})
>
s_{\mathrm{approx}}^{\mathrm{one-step}}
(Z_i;\mathcal D_{n\setminus i}+\{Z_i,Z_{n+1}^y\})
\right\}
<
(1-\alpha)(n+1).
\]
The corrected one-step set is also computed exactly as a union of intervals.
For each \(i\), the breakpoints solve
\[
\left|
(1-b)(y-d_i)-a_i r_i
\right|
=
\left|
(1-g_i)r_i-a_i(y-d_i)
\right|.
\]
Equivalently, the two candidate breakpoints are
\[
y
=
d_i
+
\frac{(1-g_i+a_i)r_i}{1-b+a_i},
\qquad
y
=
d_i
+
\frac{(a_i+g_i-1)r_i}{1-b-a_i},
\]
whenever the denominators are nonzero. The simulation sorts all such
breakpoints, evaluates the tournament acceptance condition on each induced
interval, and reports the total length of the accepted union. Coverage is
evaluated by checking the same tournament condition at \(y=Y_{n+1}\).

\subsection{Example 3: Bayesian posterior predictive density}
\label{app:bayesian_details}

For the Bayesian experiment, we use the Gaussian linear-model likelihood
\[
Y\mid X,\theta
\sim
N(X^\top\theta,\sigma_{\mathrm{lik}}^2),
\qquad
\sigma_{\mathrm{lik}}=1.
\]
The prior is
\[
\theta\sim N(\mu_0,\Sigma_0),
\qquad
\mu_0=10\cdot \mathbf 1_p,
\qquad
\Sigma_0=I_p.
\]
Since the assumed likelihood and the prior are Gaussian, conjugacy holds and we can draw the posterior samples from the closed-form Gaussian
posterior. We use \(K=100\) Monte Carlo samples.

\paragraph{Approximate.} For the uncorrected add-one-in posterior predictive approximation, draw
\[
\theta_1,\ldots,\theta_K
\sim
\pi(\cdot\mid\mathcal D_n).
\]
Equivalently, the nonconformity score used in the comparisons is
\[
s_{\mathrm{approx}}^{\mathrm{Bayes}}(z;\mathcal D_n+\{z'\})
=
-
\sum_{k=1}^K
f_{\theta_k}(y\mid x)
f_{\theta_k}(y'\mid x').
\]
Therefore, for candidate \(y\),
\[
s_{\mathrm{approx}}^{\mathrm{Bayes}}
(Z_{n+1}^y;\mathcal D_n+\{Z_{n+1}^y\})
=
-\sum_{k=1}^K
f_{\theta_k}(y\mid X_{n+1})^2,
\]
and, for $i\in[n]$,
\[
s_{\mathrm{approx}}^{\mathrm{Bayes}}
(Z_i;\mathcal D_n+\{Z_{n+1}^y\})
=
-\sum_{k=1}^K
f_{\theta_k}(Y_i\mid X_i)
f_{\theta_k}(y\mid X_{n+1}).
\]

\paragraph{Tournament-corrected.} For the tournament-corrected Bayesian method, for each \(i\in[n]\), we draw
\[
\theta_{i,1},\ldots,\theta_{i,K}
\sim
\pi(\cdot\mid\mathcal D_{n\setminus i}).
\]
For the pairwise comparison between \(Z_i\) and \(Z_{n+1}^y\), using
samples from \(\pi(\cdot\mid\mathcal D_{n\setminus i})\),
\[
s_{\mathrm{approx}}^{\mathrm{Bayes}}
(Z_{n+1}^y;\mathcal D_{n\setminus i}+\{Z_i,Z_{n+1}^y\})
=
-\sum_{k=1}^K
f_{\theta_{i,k}}(y\mid X_{n+1})^2
f_{\theta_{i,k}}(Y_i\mid X_i),
\]
and
\[
s_{\mathrm{approx}}^{\mathrm{Bayes}}
(Z_i;\mathcal D_{n\setminus i}+\{Z_i,Z_{n+1}^y\})
=
-\sum_{k=1}^K
f_{\theta_{i,k}}(Y_i\mid X_i)^2
f_{\theta_{i,k}}(y\mid X_{n+1}).
\]

All Monte Carlo sums for the Bayesian methods are evaluated on the log scale
using the log-sum-exp trick. Coverage is evaluated directly at the true response
\(Y_{n+1}\). The reported lengths are computed by an adaptive grid
approximation. For both the methods the search range is taken as $[Y_{\min }-0.02\left(Y_{\max }-Y_{\min }\right),  Y_{\max }+0.02\left(Y_{\max }-Y_{\min }\right)]$.
The adaptive search starts with a coarse grid and repeatedly refines the left
and right endpoints by a shrink factor of \(10\), until the final grid
resolution is \(0.1\). If no accepted point is found on the search range, the
length is recorded as zero. Otherwise, the length is recorded as the difference
between the estimated right and left endpoints of the accepted interval.

\section{Real data implementation details}\label{app:realdata_details}
\subsection{Rounding: diabetes data}
\label{app:rounding_diabetes_implementation}

We follow the diabetes-data implementation of \citet{lee2025leave} as closely as possible. The dataset has \(n=442\) observations and \(p=10\) predictors. To mirror their preprocessing, we normalize the full dataset by
\[
X \leftarrow \frac{X-\bar X}{s_X \sqrt p},
\qquad
Y \leftarrow \frac{Y-\bar Y}{s_Y},
\]
where \(\bar X\) and \(s_X\) denote the columnwise empirical mean and population standard deviation of the full design matrix, and \(\bar Y\) and \(s_Y\) are the empirical mean and population standard deviation of the full response vector. In \citet{lee2025leave}, the real-data experiments are repeated \(100\) times with test-set size $100$, miscoverage level \(\alpha=0.1\). We use $M=100$ grid points for rounding.

The underlying prediction model in \citet{lee2025leave}'s implementation is a no-intercept linear predictor
\[
f_{\widehat\beta}(x)=x^\top \widehat\beta,
\]
where \(\widehat\beta\) solves
\[
\widehat\beta
\in
\arg\min_{\beta\in\mathbb R^p}
\left\{
\frac{1}{2n}\sum_{i=1}^n \phi_\epsilon\!\left(Y_i-X_i^\top\beta\right)
+
\frac{\lambda}{2}\|\beta\|_2^2
\right\},
\]
with \(\phi_\epsilon\) denoting the Huber loss,
\[
\phi_\epsilon(u)
=
\begin{cases}
u^2, & |u|\le \epsilon,\\[3pt]
2\epsilon |u|-\epsilon^2, & |u|>\epsilon.
\end{cases}
\]
We use the same tuning parameters as \citet{lee2025leave},
\[
\epsilon=1,
\qquad
\lambda=2.
\]
Their public code solves this optimization problem in \texttt{cvxpy}. In our R implementation, we solve the same convex objective using \texttt{optim} with the BFGS algorithm. The nonconformity score is the absolute residual
\[
s(z;\mathcal D)=|y-f_{\widehat\beta}(x)|.
\]

For the approximate rounding method, we use the same grid-search FullCP construction as in the public code of \citet{lee2025leave}. Given a training set and a test covariate \(X_{n+1}\), define the candidate-response grid
\[
\mathcal Y_{\mathrm{grid}}
=
\left\{
y^{(1)},\dots,y^{(100)}
\right\},
\]
where the \(100\) grid points are equally spaced between
\[
\min(Y_{\mathrm{train}})-\operatorname{sd}(Y_{\mathrm{train}})
\quad\text{and}\quad
\max(Y_{\mathrm{train}})+\operatorname{sd}(Y_{\mathrm{train}}).
\]
For each candidate \(y\in\mathcal Y_{\mathrm{grid}}\), we augment the training data with \(Z_{n+1}^y=(X_{n+1},y)\), refit the robust linear model on the augmented sample, and compute the \(n+1\) absolute residual scores. We then include \(y\) whenever the test-point score is no larger than the empirical \((1-\alpha)\)-quantile of these \(n+1\) scores.

For the tournament-corrected rounding method, we use the same candidate grid \(\mathcal Y_{\mathrm{grid}}\), the same robust linear model, and the same absolute-residual score, but replace the approximate full conformal check by the corrected pairwise comparison from Section~\ref{sec:tourn_correction}. Concretely, for each candidate \(y\in\mathcal Y_{\mathrm{grid}}\) and each \(i\in[n]\), we compute
\[
s_{\mathrm{approx}}\bigl(Z_{n+1}^y;\mathcal D_{n\setminus i}+\{Z_i,Z_{n+1}^y\}\bigr)
\qquad\text{and}\qquad
s_{\mathrm{approx}}\bigl(Z_i;\mathcal D_{n\setminus i}+\{Z_i,Z_{n+1}^y\}\bigr),
\]
where the approximation is of rounding type, with \(Y_i\) rounded to its nearest grid value and \(y\) restricted to the same candidate grid.

Finally, for both methods we report empirical coverage and average prediction-set length over \(100\) random train/test splits.

\subsection{Bayes: diabetes data}
\label{app:bayes_diabetes_implementation}

We use the diabetes regression data from the \texttt{lars} package, which contains $n=442$ observations and $p=10$ covariates. For each repetition, we randomly split the data into a $70\%$ training set and a $30\%$ test set. The covariates and response are standardized within each split using the training-set mean and standard deviation, and the same covariate standardization is then applied to the test set. All reported prediction-set lengths for this experiment are on the standardized response scale.

Following the sparse-regression experiment of \citet{fong2021conformal}, we use a Bayesian lasso model. Conditional on the covariates, the likelihood is
\[
Y_i \mid X_i,\beta_0,\beta,\sigma
\sim
N(\beta_0 + X_i^\top \beta,\sigma^2).
\]
The regression coefficients are given independent Laplace priors,
\[
\beta_j \mid b \sim \mathrm{Laplace}(0,b),
\qquad j=1,\dots,p,
\]
with hyperprior
\[
b \sim \mathrm{Gamma}(1,1).
\]
For the noise scale, we use
\[
\sigma \sim \mathrm{HalfNormal}(1).
\]
The intercept is not penalized. In the Stan implementation, we use a weak Gaussian prior
\[
\beta_0 \sim N(0,10^2),
\]
which is effectively diffuse on the standardized response scale and is used for numerical stability. This is the same Bayesian lasso specification used for the conformal Bayes comparison, up to this weak regularization of the intercept.

Posterior sampling is performed using \texttt{rstan}. For each posterior fit, we run $4$ chains with $3000$ iterations per chain, of which $1000$ are warmup iterations. Thus each fit produces $2000$ post-warmup samples per chain. We use the same settings for both the full-data posterior used by the approximate Bayesian method and the leave-one-out posteriors used by the corrected method. The initialization is chosen deterministically: the coefficients are initialized at zero, the intercept at the empirical mean of the standardized training responses.

Prediction-set length is computed using the same grid-based approximation as in the conformal Bayes implementation. For each test point, we evaluate membership on a one-dimensional grid of $100$ candidate response values,
\[
\mathcal G
=
\left\{
\min(Y_{\mathrm{train}})-2,\dots,\max(Y_{\mathrm{train}})+2
\right\},
\]
where the response is on the standardized scale. If $\Delta$ denotes the grid spacing and $A\subseteq \mathcal G$ is the set of accepted grid points, then the reported length is $|A|\Delta$.

Coverage is not evaluated by rounding the true test response to the grid. Instead, for each test point we evaluate the conformal rank condition directly at the observed standardized value $Y_{n+1}$. The final result is over $50$ independent train/test splits where for each repetition, we compute the average coverage and average length over the test points. 

\section{Coverage guarantee of approximate methods under stability}\label{app:validity_approx_methods}

For completeness, in this section we show that the uncorrected approximate prediction set
\(\widehat C_{n,\alpha}^{\mathrm{approx}}\) can also satisfy a near-\((1-\alpha)\)
coverage guarantee, but under a stronger stability condition. This result is
analogous to the stability result for the naive method in
\citet{barber2021predictive}; in the deletion case,
\(s_{\mathrm{approx}}(z;\mathcal D+\{z'\})=s(z;\mathcal D)\), the approximate method is exactly the naive method. Related results also appear in \citet{ndiaye2022stable,lee2025leave}.

We define the \(2\epsilon\)-inflated version of the $\widehat  C_{n,\alpha}^{\mathrm{approx}}$ as
\begin{multline*}
    \widehat C_{n,\alpha}^{\mathrm{approx},2\epsilon}(X_{n+1})
    =
    \Biggl\{
    y:
    \sum_{i=1}^n
    \mathbf 1 \Biggl\{
    s_{\mathrm{approx}}\bigl(
    Z_{n+1}^y;\mathcal D_n+\{Z_{n+1}^y\}
    \bigr) \\ 
    > s_{\mathrm{approx}}\bigl(
    Z_i;\mathcal D_n+\{Z_{n+1}^y\}
    \bigr)+2\epsilon
    \Biggr\}
    <
    (1-\alpha)(n+1) \Biggr\}.
\end{multline*}

\begin{theorem}[Coverage of the approximate method under stronger stability]
\label{thm:validity_approx_under_stability}
Suppose that \(\mathcal D_{n+1}=\{Z_1,\dots,Z_{n+1}\}\) is exchangeable and $s_{\textnormal{approx}}$ satisfies the stability conditions in~\eqref{eq:stab_stronger} for some \(\epsilon\ge 0\) and \(\nu\in[0,1]\), for a score function \(s(z;\mathcal D\) that is symmetric in $\mathcal D$. Then
\[
\mathbb P\left(
Y_{n+1}\in
\widehat C_{n,\alpha}^{\mathrm{approx},2\epsilon}(X_{n+1})
\right)
\ge
1-\alpha-3\sqrt{\nu}.
\]
\end{theorem}

\begin{proof}
Write
\[
A_{n+1}
:=
s_{\mathrm{approx}}(Z_{n+1};\mathcal D_n+\{Z_{n+1}\}),
\qquad
A_i
:=
s_{\mathrm{approx}}(Z_i;\mathcal D_n+\{Z_{n+1}\}),
\quad i\in[n],
\]
and define the corresponding full-conformal scores
\[
S_j:=s(Z_j;\mathcal D_{n+1}),
\qquad j=1,\dots,n+1.
\]
By exchangeability of \(\mathcal D_{n+1}\) and symmetry of \(s\), the vector
\((S_1,\dots,S_{n+1})\) is exchangeable.

By definition of the \(2\epsilon\)-inflated approximate set,
\[
Y_{n+1}
\notin
\widehat C_{n,\alpha}^{\mathrm{approx},2\epsilon}(X_{n+1})
\quad
\Longleftrightarrow
\quad
\sum_{i=1}^n
\mathbf 1\left\{
A_{n+1}>A_i+2\epsilon
\right\}
\ge
(1-\alpha)(n+1).
\]

Define the stability events
\[
E_{n+1}
:=
\left\{
|A_{n+1}-S_{n+1}|\le \epsilon
\right\},
\qquad
E_i
:=
\left\{
|A_i-S_i|\le \epsilon
\right\},
\quad i\in[n].
\]
On the event \(E_{n+1}\cap E_i\), if
\[
A_{n+1}>A_i+2\epsilon,
\]
then
\[
S_{n+1}
\ge
A_{n+1}-\epsilon
>
A_i+\epsilon
\ge
S_i.
\]
Therefore, on the event \(E_{n+1}\),
\[
\mathbf 1\left\{
A_{n+1}>A_i+2\epsilon
\right\}
\le
\mathbf 1\left\{
S_{n+1}>S_i
\right\}
+
\mathbf 1\{E_i^c\}.
\]
Summing over \(i\in[n]\), we get, on \(E_{n+1}\),
\[
\sum_{i=1}^n
\mathbf 1\left\{
A_{n+1}>A_i+2\epsilon
\right\}
\le
\sum_{i=1}^n
\mathbf 1\left\{
S_{n+1}>S_i
\right\}
+
\sum_{i=1}^n
\mathbf 1\{E_i^c\}.
\]
Therefore,
\begin{multline*}
    \mathbb P\left(
Y_{n+1}
\notin
\widehat C_{n,\alpha}^{\mathrm{approx},2\epsilon}(X_{n+1})
\right) 
\le
\mathbb P(E_{n+1}^c)
+
\mathbb P\left(
\sum_{i=1}^n \mathbf 1\{E_i^c\}
>
\sqrt{\nu}(n+1)
\right) \\ +
\mathbb P\left(
\sum_{i=1}^n
\mathbf 1\left\{
S_{n+1}>S_i
\right\}
\ge
(1-\alpha-\sqrt{\nu})(n+1)
\right).
\end{multline*}
By our stability assumptions, along with exchangeability of the data, we have $\mathbb P(E_i^c)\le \nu$ for all $i$. Therefore,
\[
 \mathbb P\left(\sum_{i=1}^n \mathbf 1\{E_i^c\}
>
\sqrt{\nu}(n+1)
\right)
\leq
\frac{
\sum_{i=1}^n \mathbb P(E_i^c)
}{
\sqrt{\nu}(n+1)
} \leq \frac{n\nu}{\sqrt{\nu}(n+1)}
\leq \sqrt{\nu},\]
where the first step follows by Markov's inequality.
Since
\((S_1,\dots,S_{n+1})\) is exchangeable, for any \(t\in\mathbb R\),
\[
\mathbb P\left(
\sum_{i=1}^n
\mathbf 1\{S_{n+1}>S_i\}
\ge (1-\alpha -\sqrt{\nu})(n+1)
\right)
\le
\alpha+\sqrt{\nu}.
\]
Combining these, we obtain
\[
\mathbb P\left(
Y_{n+1}
\notin
\widehat C_{n,\alpha}^{\mathrm{approx},2\epsilon}(X_{n+1})
\right)
\le
\nu+\sqrt{\nu}+\alpha+\sqrt{\nu}
\le
\alpha+3\sqrt{\nu},
\]
and the result follows.
\end{proof}

\section{Computational speedup for tournament-corrected Bayesian approximate PPD scores}\label{app:bayesian_computation_shortcut}

In this section, we consider the question of computational cost for the Bayesian PPD method (Example 3). In particular, it is often costly to draw samples from posterior distributions; here we examine how this cost plays a role in the tournament-corrected method.

To implement the (uncorrected) approximation of \citet{fong2021conformal} (see Example 3 in Section~\ref{sec:background}), we need to sample $K$ times from the posterior of $\theta$ given the dataset $\mathcal D_n$:
\[\theta_1,\dots,\theta_K \sim \pi(\cdot\mid \mathcal D_n).\]
On the other hand, a naive implementation of the tournament-corrected method for this example (see Example 3 in Section~\ref{sec:tournament_corrected_examples}) requires sampling
\[\theta_{1,i},\dots,\theta_{K,i} \sim \pi(\cdot\mid \mathcal D_{n\setminus i}).\]
In this construction, $K$ posterior samples are drawn for \emph{each} leave-one-out dataset $\mathcal D_{n\setminus i}$. While the randomized coverage guarantee of Theorem~\ref{thm:validity_randomized} ensures that this construction leads to $\geq 1-2\alpha$ coverage, this comes at a computational cost: an $n$-fold increase in sampling cost relative to the AOI-PPD method, i.e., the uncorrected approximation.

We now develop a different strategy that allows us to retain the coverage guarantee offered by the corrected method, while avoiding the $n$-fold increase in cost, incurred by drawing independent samples from each posterior $\pi(\cdot\mid \mathcal D_{n\setminus i})$. We will now see how we can use rejection sampling to avoid this increase in cost. We will need to assume that the likelihood is lower-bounded, $\inf_\theta f_\theta(y\mid x) \geq c(z)$ for some $c(z)>0$, for each data point $z=(x,y)$.
We will draw the posterior samples as follows:
\begin{itemize}
\item[] \hspace{-.25in} Bayesian PPD with shared sampling:
\item Define $k_i=0$ for all $i\in[n]$.
\item Repeat until $k_i = K$ for all $i\in[n]$:
\begin{itemize}
    \item Sample $\theta\sim \pi(\cdot\mid \mathcal{D}_n)$
    \item 
    For each $i\in[n]$ with $k_i<K$:
    \begin{itemize}
    \item Sample $V\sim\textnormal{Unif}[0,1]$.
        \item If $V\leq \frac{c(Z_i)}{f_\theta(Y_i\mid X_i)}$, then set $k_i\leftarrow k_i+1$
    and define $\theta_{k_i,i} = \theta$.
    \end{itemize}
\end{itemize}
\end{itemize}
In other words, we use the same stream of candidate draws $\theta$, and then perform rejection sampling for each $i\in[n]$ to define posterior samples $\theta_{k,i}\sim \pi(\cdot\mid \mathcal D_{n\setminus i})$.
We now need to see how this construction can be used within the framework of our coverage results: it is not immediately clear whether this sampling scheme might be treating the data points (training + test data) in a nonsymmetric way, violating exchangeability. 

Our next observation is that we can view the scheme above in a different way. Imagine that we have access to the test point, so that we can use the entire dataset $\mathcal D_{n+1}$ for sampling. Then we proceed as follows: 
\begin{itemize}
\item[] \hspace{-.25in}  Bayesian PPD with shared sampling---alternative representation:
\item Define $k_i=0$ for all $i\in[n]$.
\item Repeat until $k_i = K$ for all $i\in[n]$:
\begin{itemize}
    \item Sample $\theta\sim \pi(\cdot\mid \mathcal{D}_{n+1})$ and $U,V_1,\dots,V_n\stackrel{\textnormal{d}}{\sim}\textnormal{Unif}[0,1]$, independently.
    \item For each $i\in[n]$ with $k_i<K$:
    \begin{itemize}
        \item 
    If $U\leq \frac{c(Z_{n+1})}{f_\theta(Y_{n+1}\mid X_{n+1})}$ and $V_i\leq \frac{c(Z_i)}{f_\theta(Y_i\mid X_i)}$, then set $k_i\leftarrow k_i+1$
    and define $\theta_{k_i,i} = \theta$.
        \end{itemize}
\end{itemize}
\end{itemize}

But sampling $\theta\sim \pi(\cdot\mid \mathcal{D}_{n+1})$, and then proceeding only if the criterion $U\leq \frac{c(Z_{n+1})}{f_\theta(Y_{n+1}\mid X_{n+1})}$ is satisfied, is exactly equivalent to simply drawing $\theta\sim \pi(\cdot\mid \mathcal{D}_n)$ (i.e., this is rejection sampling, which removes the test data point $Z_{n+1}$ from the posterior). In other words, this is exactly equivalent to the sampling scheme defined above.

Now we will see how to apply Theorem~\ref{thm:validity_randomized}. Let $\xi\sim\textnormal{Unif}[0,1]$ be a random seed. Abusing notation, for any $\xi\in[0,1]$ we will write $\xi = (\xi',\xi'',\xi''')$ to denote `splitting' the random seed into three i.i.d.\ $\textnormal{Unif}[0,1]$ random seeds (e.g., by cycling through the random digits of $\xi$). Define a deterministic function $\texttt{Sample}_\pi(\mathcal D;\xi)$ such that
\[(\theta^{(1)},\theta^{(2)},\dots)=\texttt{Sample}_\pi(\mathcal D;\xi)\]
returns an infinite stream of i.i.d.\ draws $\theta\sim\pi(\cdot\mid \mathcal D)$ when the random seed is sampled as $\xi\sim\textnormal{Unif}[0,1]$. 
Similarly define a deterministic function $\texttt{Sample}_{\textnormal{Unif}}(\xi)$ such that
\[(V^{(1)},V^{(2)},\dots)=\texttt{Sample}_{\textnormal{Unif}}(\xi)\]
returns an infinite stream of i.i.d.\ draws $V\sim\textnormal{Unif}[0,1]$ when the random seed is sampled as $\xi\sim\textnormal{Unif}[0,1]$. 
Next define $\texttt{Sample}(\mathcal D, z,z';\xi,K)$ as follows:
\begin{itemize}
    \item Define $(\theta^{(1)},\theta^{(2)},\dots)=\texttt{Sample}_\pi(\mathcal D\cup \{z,z'\};\xi')$;
    \item Define $(U^{(1)},U^{(2)},\dots)=\texttt{Sample}_{\textnormal{Unif}}(\xi'')$ and $(V^{(1)},V^{(2)},\dots)=\texttt{Sample}_{\textnormal{Unif}}(\xi''')$;
    \item Let $i_1<\dots<i_K$ be the first $K$ indices $i$ for which $U^{(i)} \leq \frac{c(z)}{f_{\theta^{(i)}}(z)}$ and $V^{(i)} \leq \frac{c(z')}{f_{\theta^{(i)}}(z')}$;
    \item Return $\theta^{(i_1)},\dots,\theta^{(i_K)}$.
\end{itemize}
In other words, this returns a sample
\[(\theta_1,\dots,\theta_K) =\texttt{Sample}(\mathcal D, z,z';\xi,K), \]
which is a deterministic function, but for a random seed $\xi\sim\textnormal{Unif}[0,1]$ this is equivalent to a rejection sampling procedure (which draws from $\pi(\cdot\mid \mathcal D\cup\{z,z'\})$ and then uses rejection sampling to approximate removing data points $z,z'$). Note that, taking $\mathcal D = \mathcal D_{n\setminus i}$ and $z=Z_{n+1},z'=Z_i$, this is exactly equivalent to sampling scheme (B) for each $i\in[n]$.

Finally, define the score
\begin{equation}\label{eqn:randomized_score_bayes}s_{\textnormal{approx}}(z;\mathcal D+\{z',z''\};\xi) = -\sum_{k=1}^K f_{\theta_k}(y\mid x) \cdot f_{\theta_k}(y'\mid x') \cdot f_{\theta_k}(y''\mid x'')  \end{equation}
where
\[(\theta_1,\dots,\theta_K) = \texttt{Sample}(\mathcal D,z',z'';\xi,K).\]

Now we apply this to the augmented dataset $\mathcal D_{n+1}$. Let $\xi', \xi''_1,\dots,\xi''_{n+1}\stackrel{\textnormal{iid}}{\sim}\textnormal{Unif}[0,1]$, and define
\[\xi_{ij} = (\xi', \xi''_i, \xi''_j)\]
for $i,j\in[n+1]$. We can verify that the tournament-corrected method (with shared sampling, as defined above) is therefore equivalent to using the randomized score function~\eqref{eqn:randomized_score_bayes} with this random seed construction.
Since this randomized score function satisfies the conditions~\eqref{eqn:random_seed} required by Theorem~\ref{thm:validity_randomized}, this verifies theoretical validity.


\end{document}